\documentclass[preprint,12pt]{elsarticle}

\usepackage{amssymb}

\usepackage{algorithm}
\usepackage{array}
\usepackage{amsmath}
\usepackage{algpseudocode}
\usepackage{blindtext}
\usepackage{textcomp}
\usepackage{subcaption}
\usepackage{amsthm}
\usepackage{xcolor}

\usepackage{geometry}

\newtheorem{thm}{Theorem} 
\newtheorem{lem}[thm]{Lemma}
\newdefinition{rmk}{Remark} 
\newproof{pf}{Proof} 
\newproof{pot}{ProofofTheorem\ref{thm2}}
\newtheorem{corollary}{Corollary}

\DeclareUnicodeCharacter{0301}{}
\DeclareUnicodeCharacter{2113}{$\ell$}

\makeatletter
\def\ps@pprintTitle{%
 \let\@oddhead\@empty
 \let\@evenhead\@empty
 \def\@oddfoot{\centerline{\thepage}}%
 \let\@evenfoot\@oddfoot}
\makeatother

\begin{document}

\begin{frontmatter}

\title{Parameterized Linear Time Transitive Closure}

\author{Giorgos Kritikakis\corref{cor1}} \fnref{aff, fn1} \ead{georgecretek@gmail.com}
\author{ Ioannis G Tollis \fnref{aff} }  \ead{tollis@csd.uoc.gr}

\cortext[cor1]{Corresponding author}
\fntext[fn1]{Giorgos Kritikakis is currently employed by Tom Sawyer Software.}
\affiliation[aff]{organization={University of Crete},
            addressline={Voutes Campus}, 
            city={Heraklion},
            postcode={70013}, 
            country={Greece}}

\begin{abstract}

Inquiries such as whether a task A depends on a task B, whether an author A has been influenced by a paper B, whether a certain protein is associated with a specific biological process or molecular function, or whether class A inherits from class B, are just a few examples of inquiries that can be modeled as reachability queries on a network (Directed Graph). Digital systems answer myriad such inquiries every day.

In this paper, we discuss the transitive closure problem. We focus on applicable solutions that enable us to answer
queries fast, in constant time, and can serve in real-world applications. In contrast to the majority of research on 
this topic that revolves around the construction of a two-dimensional adjacency matrix, we present an approach 
that builds a reachability indexing scheme. 
This scheme enables us to answer queries in constant time and can be built in parameterized linear time. In addition, it captures a compressed data structure. Our approach and algorithms are validated by extensive
experiments that shed light on the factors that play a key role in this problem.
To stress the efficiency of this solution and demonstrate the potential to apply our approach to important problems, we use it to speed up Fulkerson's method for finding the width of a DAG. Our results challenge the prevailing belief, reiterated over the last thirty years, regarding the efficiency of this method.

Our approach is based on the concept of chain decomposition. Before we delve into its description, we introduce,
analyze, and utilize a chain decomposition algorithm. Furthermore, we explore how chain decomposition can facilitate transitive closure solutions introducing a general-purpose linear time reduction technique that removes a large subset of transitive edges given any chain decomposition.
\end{abstract}


\end{frontmatter}


\section{Introduction}
\label{sec:Intro}
Computing the transitive closure or reachability information of a directed graph is fundamental in computer science and is the basic step in many applications.
Formally, given a directed graph $G = (V, E)$, the transitive closure of $G$, denoted as $G$*, is a graph ($V, E$*) such that $E$* contains all edges in $E$, and for any pair of vertices $u$, $v \in V$, if there exists a directed path from $u$ to $v$ in $G$, then there is a directed edge from $u$ to $v$ in $E$*.
 An edge $(v_1,v_2)$ of a DAG $G$ is transitive if there is a path longer than one edge that connects $v_1$ and $v_2$.
 Given a directed graph with cycles, we can find the strongly connected components (SCC) in linear time and collapse all vertices of a SCC into a supernode. Hence, any reachability query can be reduced to a query in the resulting \emph{Directed Acyclic Graph} (DAG). Additionally, DAGs are very important in many applications in several areas of research and business because they often represent hierarchical relationships between objects in a structure.  Any DAG can be decomposed into vertex disjoint \emph{paths} or \emph{chains}. In a path, every vertex is connected to its
 successor by an edge, while in a chain any vertex is connected to its successor by a directed path, which may be an edge.  A \emph{path/chain decomposition} is a set of vertex disjoint paths/chains that cover all the vertices of a DAG (see Figure~\ref{fig:Decomposition}). Two chains (or paths) $c_1$ and $c_2$ can be merged into a new chain if the last vertex of $c_1$ can reach the first vertex of $c_2$, or if the last vertex of $c_2$ can reach the first vertex of $c_1$. The process of merging two or more paths or chains into a new chain is called path or chain \emph{concatenation}.

The \emph{width} of a DAG $G=(V, E)$ is the maximum number of mutually unreachable vertices of $G$~\cite{DILWORTH}.
An \emph{optimum chain decomposition} of a DAG $G$ contains the minimum number of chains, $k$, which is equal to the width of $G$.  
Due to the multitude of applications, there are several algorithms to find a chain decomposition of a DAG, see for example~\cite{Jagadish:1990:CTM:99935.99944, Fulkerson, chendag, makinen2019sparse, caceres2021linear,caceres2022sparsifying, kogan2022beating, van2021minimum}. Some of them find the optimum and some are heuristics.  Generally speaking, the algorithms that compute the optimum take more than linear time and use flow techniques which are often heavy and complicated to implement.  On the other hand, for several practical applications, it is not necessary to compute an optimum chain decomposition. The computation of an efficient chain decomposition of a DAG has many applications in several areas including bioinformatics \cite{bonizzoni2007linear,gramm2007haplotyping}, evolutionary computation \cite{jaskowski2011formal}, databases \cite{Jagadish:1990:CTM:99935.99944}, graph drawing \cite{ortali2018algorithms,NewFrHierDrawings,lionakis2020algorithms}, distributed systems \cite{ikiz2006efficient,tomlinson1997monitoring}
.

In Section~\ref{sec:PathDec}, we review some basic path and chain decomposition approaches, 
and we introduce a new efficient way to create a chain decomposition. Our approach creates a number of chains that is very close to the optimum. We analyse and utilize this algorithm to accelerate transitive closure solutions.
Furthermore, we present experimental results that show that the width of a DAG grows as the graph becomes denser.
In fact, we compute the width for four different random DAG models in order to obtain insights about its growth as the size of the graphs grows. 
An interesting observation is that the width of DAGs created using the Erdős-Rényi (ER) random graph model is proportional to $\frac{\mbox{number of nodes}}{\mbox{average degree}}$. 
Next, in Section \ref{sec:HierTrans}, we show that $|E_{red}| \leq width*|V|$, and describe how to remove a significant subset of transitive edges, $E_{tr}'$, in linear time, in order to bound $|E-E_{tr}'|$ by $O(k*V)$ given any path/chain decomposition of size $k$. Clearly, this can boost many known transitive closure solutions. Section \ref{sec:IndexingScheme} demonstrates how to build an indexing scheme that implicitly contains the transitive closure and we report experimental results that shed light on the interplay of width, $E_{red}$, and $E_{tr}$.

We consider reachability mainly for the static case, i.e., when the graph does not change.
The question of whether an arbitrary vertex
$v$ can reach another arbitrary vertex $u$ can be answered in linear time by running a breadth-first or depth-first search from $v$, or it can be answered in constant time after a reachability indexing scheme or transitive closure of the graph has been computed. 
The transitive closure of a graph with $|V|$ vertices and $|E|$ edges can be computed in $O(|V|*|E|)$ time by starting a breadth-first or depth-first search from each vertex. Alternatively, one can use the Floyd-Warshall algorithm~\cite{floyd1962algorithm} which runs in $O(|V|^3)$, or solutions based on matrix multiplication~\cite{strassen1969gaussian}.
Currently, the best-known bound on the asymptotic complexity of a matrix multiplication algorithm
$O(|V|^{2.3728596})$ time~\cite{alman2021refined}.  
An algorithm with complexity $O(|V|^{2.37188})$ was very recently announced in a preprint~\cite{duan2023faster}.  
However, these and other similar improvements to Strassen's Algorithm are not used in practice because they are complicated and the constant coefficient hidden by the notation is extremely large. 
Here we focus on computing a reachability indexing scheme in parameterized linear time.
Notice that we do not explicitly compute the transitive closure matrix of a DAG.  The matrix can be easily computed from the reachability indexing scheme in $O(|V|^2)$ time (i.e., constant time per entry).  

\begin{figure}[hbt!]
     \centering
     \begin{subfigure}[b]{0.19\textwidth}
         \centering
         \includegraphics[width=\textwidth]{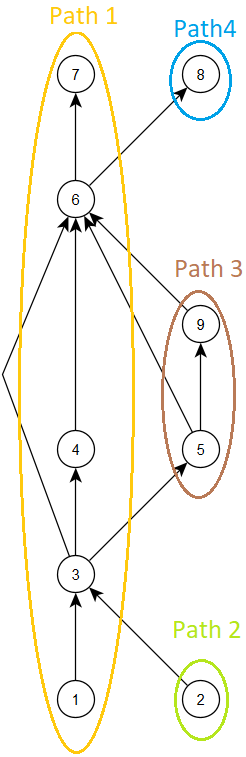}
         \caption{A path decomposition of a graph consisting of 4 paths.}
     \end{subfigure}
   \hspace{2cm}
     \begin{subfigure}[b]{0.19\textwidth}
         \centering
         \includegraphics[scale=0.55]{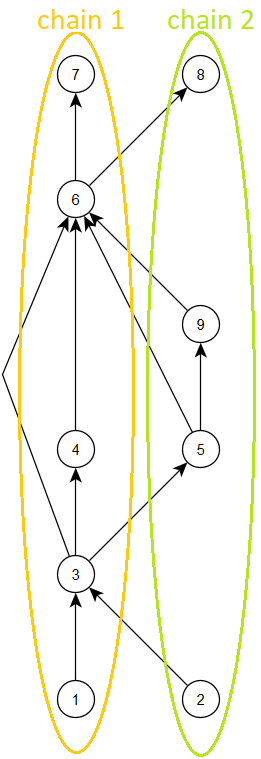}
         \caption{A chain decomposition of a graph consisting of 2 chains.}
     \end{subfigure}
             \caption{Path and chain decomposition of an example graph.}
        \label{fig:Decomposition}
\end{figure}

In this paper we present a practical algorithm to compute a reachability indexing scheme (or the transitive closure information) of a DAG $G=(V, E)$, utilizing a given path/chain decomposition. The scheme can be computed in parameterized linear time, where the parameter is the number of paths/chains, $k_c$, in the given decomposition.  The scheme can answer any reachability query in constant time.  Let $E_{tr}$, $E_{tr}\subset E$, denote the set of transitive edges and $E_{red}$, $E_{red}=E-E_{tr}$, denote the set of non-transitive edges of $G$. We show that $|E_{red}|\leq width*|V|$ and that we can compute a substantially large subset of $E_{tr}$ in linear time (see Section~ \ref{sec:HierTrans}). This implies that any DAG can be reduced to a smaller DAG that has the same transitive closure in linear time.  Consequently, several hybrid reachability algorithms will run much faster in practice.
The time complexity to produce the reachability indexing scheme is $O(|E_{tr}|+k_c*|E_{red}|) = O(|E_{tr}|+k_c*width*|V|) = O(|E_{tr}|+k_c^2*|V|)$, and its space complexity is $O(k_c*|V|)$ (see Section~\ref{sec:IndexingScheme}).  
Our experimental results reveal the practical efficiency of this approach.
In fact, the results show that our method is substantially better in practice than the theoretical bounds imply, indicating that path/chain decomposition algorithms can be used to solve the transitive closure (TC) problem. Clearly, given a reachability indexing scheme the TC matrix can be computed in $O(|V|^2)$ time.

\section{Width of a DAG and Decomposition into Paths/Chains}  \label{sec:PathDec}

In this section, we briefly describe some categories of path and chain decomposition techniques, we introduce a new chain decomposition heuristic and show experimental results. The experiments explore the practical efficiency of our chain decomposition approach and 
the behavior of the width in different graph models.
We focus on fast and practical path/chain decomposition heuristics. There are two categories of path decomposition algorithms, Node Order Heuristic, and Chain Order Heuristic, see ~\cite{Jagadish:1990:CTM:99935.99944}.
The first constructs the paths one by one, while the second creates the paths in parallel. The chain-order heuristic starts from a vertex and extends the path to the extent possible. The path ends when no more unused immediate successors can be found. The node-order heuristic examines each vertex (node) and assigns it to an existing path. If no such path exists, then a new path is created for the vertex. We assume that the vertices are topologically ordered since tracing the vertices in topological order is highly beneficial for creating compact decompositions.

In addition to path-decomposition algorithm categorization, Jagadish in~\cite{Jagadish:1990:CTM:99935.99944} describes some chain decomposition heuristics. Those heuristics run in $O(|V|^2)$ time using a pre-computed transitive closure, which is not linear, and we will not discuss them further. In this work, we do the opposite: we first compute an efficient chain decomposition very fast and use it to construct transitive closure solutions.

More precisely our algorithm decomposes a DAG into $k_c$ chains in $O({\mid}E{\mid}+c*l)$ time, where $c$ is the number of path concatenations, and $l$ is the length of a longest path of the DAG (without requiring any precomputation of the transitive closure). One may argue that the worst-case time complexity of this technique is 
$O(|V|^2)$, since $c<|V|$ and $l<|V|$. However, we show theoretically and experimentally that the factor $c*l$ is rarely larger than the number of edges $({\mid}E{\mid})$. Hence, it does not only have the tightest theoretical bound but also behaves purely like a linear time algorithm.

Our approach is a variation of the Node Order heuristic combined with a novel chain concatenation technique.

\subsection{Path/Chain Concatenation}
We present now a key component of our approach, a path concatenation technique that takes as input any path decomposition of a DAG $G=(V,E)$ and constructs a chain decomposition by performing repeated path concatenations. If it performs $c$ path concatenations and $l$ is the length of a longest path of the graph, then it requires $O({\mid}E{\mid}+c*l)$ time. In fact, the actual cost for every successful concatenation is the path between the two concatenated paths.  Hence the total time for all concatenations is bounded by $(c*l)$. 
In order to apply our path concatenation algorithm, we first compute a path decomposition of the graph. 
We can use any linear-time path decomposition Algorithm. After performing the concatenation technique, no further concatenations remain.

As we have already discussed, two chains (or paths) $c_1$ and $c_2$ can be merged into a new chain if the last vertex of $c_1$ can reach the first vertex of $c_2$, or if the last vertex of $c_2$ can reach the first vertex of $c_1$. The process of merging two or more chains into a new chain is called a path or chain concatenation. To reduce the number of chains of any given chain decomposition, we search for possible concatenations and merge the chains we find. Searching for a concatenation implies that we are searching for a path between two chains. 
We can start searching from the first vertex of a chain looking for the last vertex of another chain, or from the last vertex of a chain looking for the first vertex of another chain.

Given a DAG $G=(V, E)$ and a path decomposition $D_p$ that contains $k_p$ paths we will build a chain decomposition of $G$ that contains $k_c$ chains in $O({\mid}E{\mid}+(k_p-k_c)*l)$ time, where $l$ is the length of a longest path of $G$. This is accomplished by performing path/chain concatenations. Since each path/chain concatenation reduces the number of chains by one, the total number of such concatenations is $(k_p-k_c)$. Since paths are by definition chains of a special structure, and we start by concatenating paths into chains, we use the more general term chain in our algorithms.

For every path in $D_p$ we start a reversed DFS lookup function from the first vertex of a chain, looking for the last vertex of another chain traversing the edges backward. The reversed DFS lookup function is the well-known depth-first search graph traversal for path finding. If the reversed DFS lookup function detects the last vertex of a chain, then it concatenates the chains. If we simply perform the above, as described, then the algorithm will run in $O(k_p*{\mid}E{\mid})$ since we will run $k_p$ reversed DFS searches. However, every reversed DFS search can benefit from the previous reversed DFS results. A reversed DFS for path finding returns the path between the source vertex and the target vertex, which in our case, is the path between the first vertex of a chain and the last vertex of another chain. Hence, every execution that goes through a set of vertices $V_i$ it splits them into two vertex disjoint sets, $R_i$ and $P_i$.  $P_i$ contains the vertices of the path from the source vertex to the destination vertex and $R_i$ contains every vertex in $V_i-P_i$. If no path is found then $V_i=R_i$ and $P_i=\emptyset$.

Notice that every vertex in set $R_i$ is not the last vertex of a chain. If it were then it would belong to $P_i$ and not to $R_i$. Similarly, for every vertex in $R_i$, all its predecessors are in $R_i$ too. Hence, if a forthcoming reversed DFS lookup function meets a vertex of $R_i$, there is no reason to proceed with its predecessors.

\begin{algorithm}[!ht]
\caption{Concatenation}\label{alg:conc}
\begin{algorithmic}[1]
\Procedure{Concatenation}{$G,D$}\newline
 \textbf{INPUT:} A DAG $G=(V,E)$, and a path decomposition D of G \newline
 \textbf{OUTPUT:} A chain decomposition of G
\For{\textbf{each path:} $p_i \in D$}
\State $f_i \gets \mbox{ first vertex of } p_i$
\State $(R_i,P_i) \gets \mbox{reversed\_DFS\_lookup}(G,f_i)$
\If{$P_i\neq \emptyset$}

\State $l_i \gets \mbox{ destination vertex of } P_i$
\Comment{Last vertex of a path}
\State \textbf{Concatenate\_Paths(} $l_i$, $f_i$ \textbf{)} 
\EndIf
\State $G \gets G \setminus R_i$
\EndFor
\EndProcedure
\end{algorithmic}
\end{algorithm}

Algorithm \ref{alg:conc} shows our path/chain concatenation technique. Observe that the reversed DFS lookup function is invoked for every starting vertex of a path.
Every reversed DFS lookup function goes through the set $R_i$ and the set $P_i$, examining the nodes and their incident edges. $P_i$ is the path from the first vertex of a chain to the last vertex of another chain. The set $R_i$ contains all of the vertices the function went through except the vertices of $P_i$.  Hence we have the following theorem:

\begin{thm}
The time complexity of Algorithm \ref{alg:conc} is $O({\mid}E{\mid}+(k_p-k_c)*l)$, where $l$ is the length of a longest path in $G$.
\end{thm}

\begin{proof}
Assume that we have $k_p$ paths. We call the reversed\_DFS\_lookup function $k_p$ times. Hence, we have $(R_i, P_i)$ sets, $0\leq i <k_p$.
In every loop, we delete the vertices of $R_i$. Hence, $R_i\cap R_j=\emptyset$, $0\leq i,j<k_p$ and $i\neq j$. 
We conclude that $\bigcup\limits_{i=0}^{k_p-1} R_i \subseteq V$ and $\sum_{i=0}^{k_p-1}{\mid}R_i{\mid}\leq{\mid}V{\mid}$.

A path/chain $P_i$, $0\leq i < k_p$, is not empty if and only if a concatenation has occurred. Hence, $\sum_{i=0}^{k_p-1}{\mid}P_i{\mid}\leq c*l$ where $c$ is the number of concatenations and $l$ is the longest path of the graph. 
Since every concatenation reduces the number of paths/chains by one, we have that $c = k_p-k_c$.
\end{proof}

Please notice that according to the previous proof the actual time complexity of Algorithm \ref{alg:conc} is $O({\mid}E{\mid}+ \sum_{i=0}^{k_p-1}{\mid}P_i{\mid})$ which in practice is expected to be significantly better than $O({\mid}E{\mid}+c*l)$. 
Indeed, this is confirmed by our experimental results. We ran experiments on Erdős-Rényi (ER) model graphs of size 10k, 20k, 40k, 80k, and 160k vertices and average degree 10; the run times were 9, 34, 99, 228, and 538 milliseconds, respectively, which shows, as expected, that the execution time follows a linear function.

\subsection{A Better Chain Decomposition Heuristic} \label{section H3}

Algorithm \ref{alg:conc} describes how to produce a chain decomposition of a DAG $G$ by applying a path/chain concatenation technique and works for any given path/chain decomposition of $G$. As described, the chain concatenation is a post-processing step applied to any given decomposition. 

Algorithm \ref{alg:NH_conc}, is a Node Order variation that deploys our novel concatenation technique and some greedy choices. Crucial is that it does not separate the steps of path decomposition and chain concatenation. As you see, the block of the if-statement of line 10 is run online instead of running Algorithm~\ref{alg:conc} as a post-processing step. In other words, if we do not find an immediate predecessor, we search all predecessors using the reversed\_DFS\_lookup function. 
The reversed\_DFS\_lookup function is applied in real-time if the algorithm does not find an immediate predecessor that is the last vertex of a chain.

{Furthermore, Algorithm \ref{alg:NH_conc} adds two extra greedy choices: (a) when we visit a vertex of out-degree 1, we immediately add its unique immediate successor to its path/chain, and (b) instead of searching for the first available immediate predecessor (that is the last vertex of a path), we choose an available vertex with the lowest out-degree. These steps have a secondary effect on the decomposition outcome reducing its size by a few chains. Still, we keep them since they do not affect the theoretical or practical efficiency.

\begin{algorithm}[hbt!]
\caption{Chain Decomposition (NH conc.)}\label{alg:NH_conc}
\begin{algorithmic}[1]
\Procedure{New Heuristic with concatenation}{$G,T$}\newline
 \textbf{INPUT:} A DAG $G=(V,E)$, and a topological sorting $T(v_1,...,v_i,...,v_N)$ of G \newline
 \textbf{OUTPUT:} A chain decomposition of G
 \State $K \gets \emptyset$ \Comment{Set of chains}
\For{ \textbf{every vertex }$v_i \in T$ \textbf{in ascending topological order}}
\State $\mbox{Chain }C$
\If{$v_i \mbox{ is assigned to a chain}$}
\State $C \gets v_i\mbox{'s chain}$  \Comment{C is a pointer to a chain}
\ElsIf{$v_i\mbox{ is not assigned to a chain}$}
\State $l_i\gets\mbox{choose the immediate predecessor with the lowest outdegree }$
\State $\mbox{\hspace*{0.8cm}that is the last vertex of a chain}$
\If{$l_i = \mbox{null}$}
\State $(R_i,P_i) \gets \mbox{reverse\_DFS\_lookup(}G,v_i\mbox{)}$
\If{$P_i\neq \emptyset$}
\State $l_i \gets \mbox{ destination vertex of } P_i$
\EndIf
\State $G \gets G \setminus R_i$
\EndIf
\If{$l_i\neq \mbox{null}$}
\State $C \gets \mbox{chain indicated by }l_i$
\State \textbf{add} $v_i$ \textbf{to} $C$
\Else
\State  $C \gets \mbox{new Chain}()$
\State \textbf{add} $v_i$ \textbf{to} $C$ %
\EndIf
\State \textbf{add} $C$ \textbf{to} $K$
\EndIf
\If{there is an immediate successor $s_i$ of $v_i$ with in-degree 1}
\State \textbf{add} $s_i$ \textbf{ to } $C$ 
\EndIf

\EndFor
\State \textbf{return} $K$
\EndProcedure
\end{algorithmic}
\end{algorithm}

\subsection{DAG Decomposition: Experimental Results}
\label{sec:ExperimentsChains}

In this section we present extensive experimental results on graphs, most of which were created by NetworkX~\cite{networkx}. We use three different random graph generator models: Erdős-Rényi~\cite{erdHos1959renyi}, Barabasi-Albert~\cite{barabasi1999emergence}, and Watts-Strogatz~\cite{watts1998collective} models.
The generated graphs are made acyclic, by orienting all edges from low to high ID number, see the documentation of networkx~\cite{networkx} for more information about the generators.  Additionally, we use the Path‑Based DAG Model that was introduced in~\cite{lionakis2021constant} and is especially designed for DAGs with a predefined number of randomly created paths. For every model, we created 12 types of graphs: Six types of 5000 nodes and six types of 10000 nodes, both with average degree 5, 10, 20, 40, 80, and 160. We used different average degrees in order to have results for various sizes and density. 

We ran the heuristics on multiple copies of graphs and examined the performance of the heuristics in terms of the number of chains in the produced chain decompositions. 
All experiments were conducted on a simple laptop PC (Intel(R) Core(TM) i5-6200U CPU, with 8 GB of main memory). Our algorithms have been developed as stand-alone java programs.
We observed that the graphs generated by the same generator with the same parameters have small width deviation. For example, the percentage of deviation on ER and Path-Based model is about 5\% and for the BA model is less than 10\%. The width deviation of the graphs in the WS model is a bit higher, but this is expected since the width of these graphs is significantly smaller.
\newline
\newline
\textbf{Random Graph Generators:}
\begin{itemize}
  \item \textbf{Erdős-Rényi (ER) model} \cite{erdHos1959renyi}: The generator returns a random graph $G_{n,p}$, where $n$ is the number of nodes and every edge is formed with probability $p$.
  \item \textbf{Barabási–Albert (BA) model} \cite{barabasi1999emergence}: preferential attachment model: A graph of $n$ nodes is grown by attaching new nodes each with $m$ edges that are preferentially attached to existing nodes with high degree. The factors $n$ and $m$ are parameters to the generator.
  \item \textbf{Watts–Strogatz (WS) model} \cite{watts1998collective}: small-world graphs: First it creates a ring over $n$ nodes. Then each node in the ring is joined to its $k$ nearest neighbors. Then shortcuts are created by replacing some edges as follows: for each edge $(u,v)$ in the underlying “$n$-ring with $k$ nearest neighbors” with probability $b$ replace it with a new edge $(u,w)$ with uniformly random choice of an existing node $w$. The factors $n,k,b$ are the parameters of the generator.
  \item \textbf{Path‑Based DAG (PB) model} \cite{lionakis2021constant}:  In this model, graphs are randomly generated based on a number of predefined but randomly created paths.
\end{itemize}

We compute the minimum set of chains using the method of Fulkerson~\cite{Fulkerson}, which in brief consists of the following steps:  1) Construct the transitive closure $G^*(V,E')$ of $G$, where $V=\{v_1,...,v_n\}$. 2) Construct a bipartite graph $B$ with partitions $(V_1, V_2)$, where $V1=\{x_1, x_2,..., x_n\},$ $ V2=\{y_1, y_2, ..., y_n\} $. An edge $(x_i,y_j)$ is formed whenever $(v_i,v_j)\in E'$.  3) Find a maximal matching $M$ of $B$. The width of the graph is $n-{\mid}M{\mid}$. In order to construct the minimum set of chains, for any two edges $e_1,e_2\in M$, if $e_1=(x_i,y_t)$ and $e_2=(x_t,y_j)$ then connect $e_1$ to $e_2$.

The aim of our experiments is twofold: (a) to understand the behavior of the width of DAGs created in different models, and (b) to compare the behavior of our heuristic used on graphs of these models.
Table~\ref{table:ChainsWidth} shows the width and the number of chains created by our algorithm for all graphs of 5000 and 10000 nodes.
Observe that our chain decomposition heuristic, Algorithm~\ref{alg:NH_conc}, computes a chain decomposition that is very close to the optimum (width).
\begin{itemize}
  \item \textbf{NH conc.}: Chain decomposition using Algorithm \ref{alg:NH_conc}
  \item \textbf{Width}: The width of the graph (computed by Fulkerson's method). 
\end{itemize}

\vspace{10pt}
\textbf{Understanding the width in DAGS:}
In order to understand the behavior of the width on DAGs of these different models we observe: (i)~the BA model produces graphs with a larger width than ER, and (ii)~the ER model creates graphs with a larger width than WS.  For the WS model, we created two sets of graphs: The first has probability $b$ equals 0.9 and the second 0.3. Clearly, if the probability $b$ of rewiring an edge is 0, the width would be one, since the generator initially creates a path that goes through all vertices. As the rewiring probability $b$ grows, the width grows. This is the reason we choose a low and a high rewiring probability. Figures \ref{fig:width5000} and \ref{fig:width10000} demonstrate the behavior of the width for each model on the graphs of 5000 and 10000 nodes. 
Please notice that in almost all model graphs (except WS 0.3) the width of a DAG decreases fast as the density of the DAG increases.  Additionally,  it is interesting to observe here that the width of the ER model graphs is proportional to $\frac{\mbox{number of nodes}}{\mbox{average degree}}$. 

Finally, we provide a comparison of the width of the ER and the Path-Based model graphs, shown in Figure~\ref{fig:Chart_ErdosPBComp10000} for graphs of 10000 nodes and varying average degrees. We show this separately since the important details would not be visible in Figure~\ref{fig:width10000}. It is interesting to note that for sparser PB graphs the width is very close to the number of predefined (but randomly created paths) whereas the width of the ER graphs is very high.  However, as the graphs become denser the width for both models seems to eventually converge. 
\\ \\

\textbf{Number of chains versus the width in DAGS:}
In order to compare the number of chains produced by our approach on the graphs of these models, besides the aforementioned tables, we also created several figures.  Namely, we visualize how close is the number of chains produced by our heuristics to the width. In particular Figures \ref{fig:Chart_WidthBarabasiComp10000}, \ref{fig:Chart_WidthErdosComp10000}, \ref{fig:Chart_WidthWatts09Comp10000}, and \ref{fig:Chart_WidthPBMComp10000}, show how close is the number of chains produced by our technique (i.e., the blue line) to the width (i.e., red line) for ER, BA, WS, and PB models. 

Our heuristics run very fast, as expected, in just a few milliseconds.
Thus it is not interesting to elaborate further on their running time in this subsection. 
In the following sections, we present partial run-time results that are obtained for computing an indexing scheme (see Tables~\ref{table:IndexingSchemeResults5000}, and~\ref{table:IndexingSchemeResults10000}).  These results indicate that the running time of the heuristics is indeed very low.

\begin{table}
    \centering
    \small\addtolength{\tabcolsep}{-1pt}
    \begin{tabular}{|c||c|c|c|c|c|c||c|c|c|c|c|c|} 
        \hline
         & \multicolumn{6}{c||}{$|V|=5000$} & \multicolumn{6}{c|}{$|V|=10000$} \\
        \hline
        Av. Degree & 5 & 10 & 20 & 40 & 80 & 160 & 5 & 10 & 20 & 40 & 80 & 160 \\
        \hline
        \hline
        \multicolumn{13}{|c|}{BA} \\
        \hline
        NH conc & 1630 & 1055 & 664 & 355 & 207 & 163 & 3341 & 2159 & 1264 & 752 & 400 & 228 \\
        Width & 1593 & 1018 & 623 & 320 & 187 & 163 & 3282 & 2066 & 1172 & 678 & 351 & 198 \\
        \hline
        \multicolumn{13}{|c|}{ER} \\
        \hline
        NH conc & 923 & 492 & 252 & 139 & 70 & 38 & 1837 & 1003 & 516 & 271 & 139 & 72 \\
        Width & 785 & 403 & 217 & 110 & 56 & 33 & 1561 & 802 & 409 & 219 & 110 & 58 \\
        \hline
        \multicolumn{13}{|c|}{WS, b=0.9} \\
        \hline
        NH conc & 687 & 212 & 60 & 25 & 20 & 17 & 1332 & 447 & 100 & 29 & 24 & 22 \\
        Width & 560 & 187 & 54 & 22 & 17 & 15 & 1101 & 378 & 93 & 27 & 20 & 18 \\
        \hline
        \multicolumn{13}{|c|}{WS, b=0.3} \\
        \hline
        NH conc & 9 & 4 & 4 & 5 & 4 & 5 & 12 & 4 & 4 & 4 & 4 & 4 \\
        Width & 9 & 4 & 4 & 4 & 4 & 4 & 12 & 4 & 4 & 4 & 4 & 4 \\
        \hline
        \multicolumn{13}{|c|}{PB, Paths=100} \\
        \hline
        NH conc & 86 & 101 & 107 & 93 & 73 & 51 & 125 & 141 & 153 & 142 & 120 & 89 \\
        Width & 70 & 70 & 70 & 68 & 58 & 30 & 100 & 100 & 100 & 99 & 90 & 47 \\
        \hline
    \end{tabular}
    \caption{The number of chains produced by Algorithm~\ref{alg:NH_conc}, compared to the optimum, for graphs with 5000 and 10000 nodes.}
    \label{table:ChainsWidth}
\end{table}

\begin{figure}
     \centering
     \begin{subfigure}[b]{0.8\textwidth}
         \centering
         \includegraphics[scale=0.8]{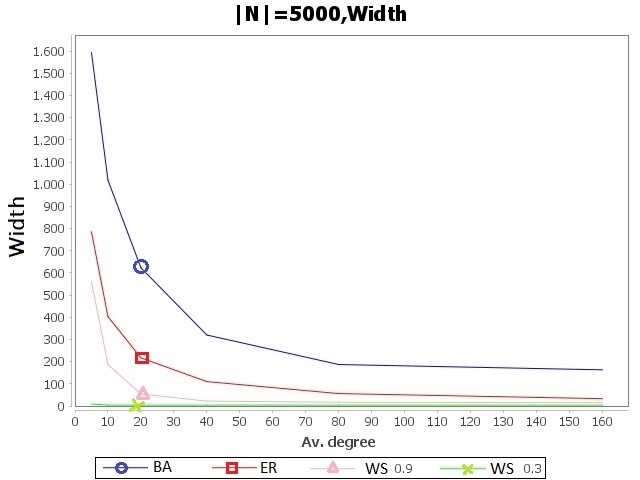}
         \caption{The width curve on graphs of 5000 nodes.}
         \label{fig:width5000}
     \end{subfigure}
     \hfill
     \begin{subfigure}[b]{0.8\textwidth}
         \centering
         \includegraphics[scale=0.8]{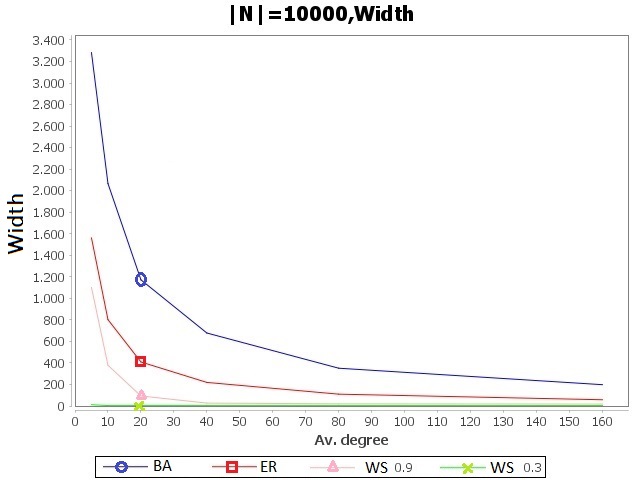}
         \caption{The width curve on graphs of 10000 nodes.}
         \label{fig:width10000}
     \end{subfigure}
    \caption{The width curve on graphs of 5000 and 10000 nodes using three different models.
}
        \label{fig:width_comp}
\end{figure}

\begin{figure}[!ht]
\centering
 \includegraphics[scale=0.8]{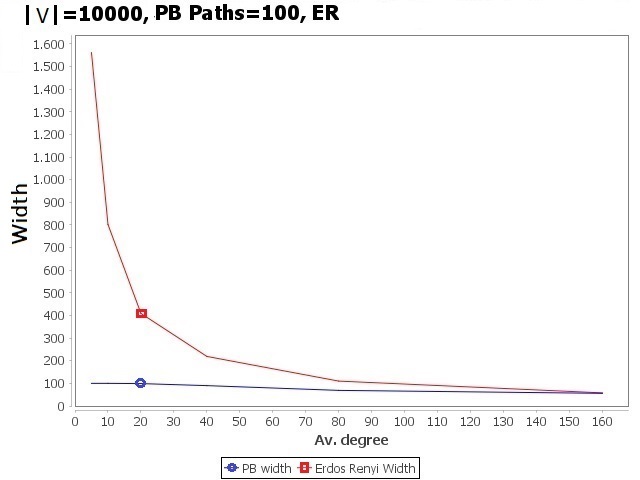}
\caption{A comparison of width between ER model and PB model.}
\label{fig:Chart_ErdosPBComp10000}
\end{figure}

\begin{figure}[!ht]
\centering
\includegraphics[scale=0.8]{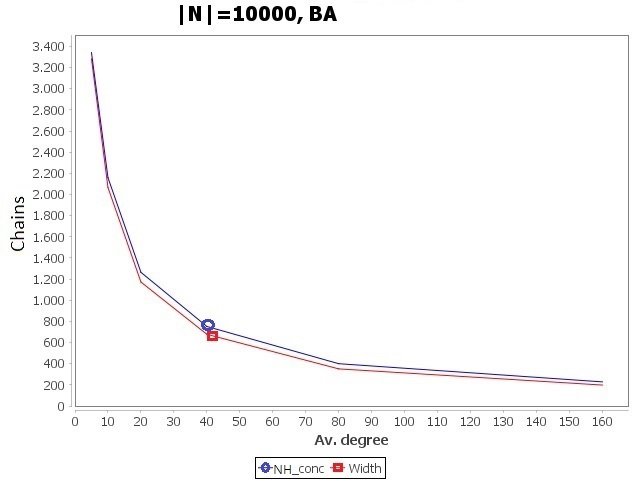}
\caption{The efficiency of the chain decomposition algorithm in the BA model.}
\label{fig:Chart_WidthBarabasiComp10000}
\end{figure}

\begin{figure}[!ht]
\centering
\includegraphics[scale=0.8]{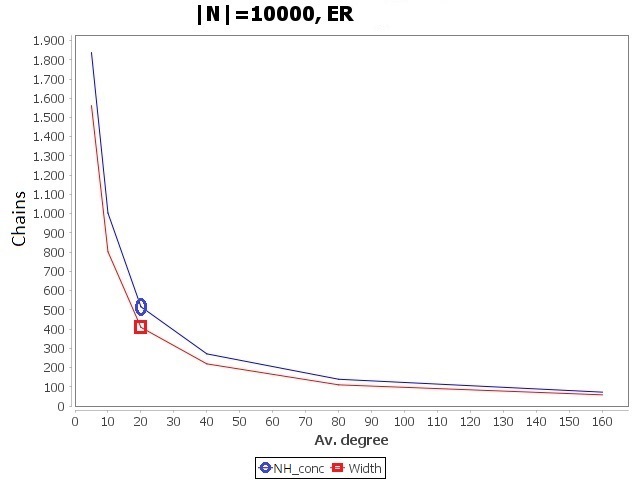}
\caption{The efficiency of our chain decomposition algorithm in the ER model.}
\label{fig:Chart_WidthErdosComp10000}
\end{figure}

\begin{figure}[!ht]
\centering
 \includegraphics[scale=0.8]{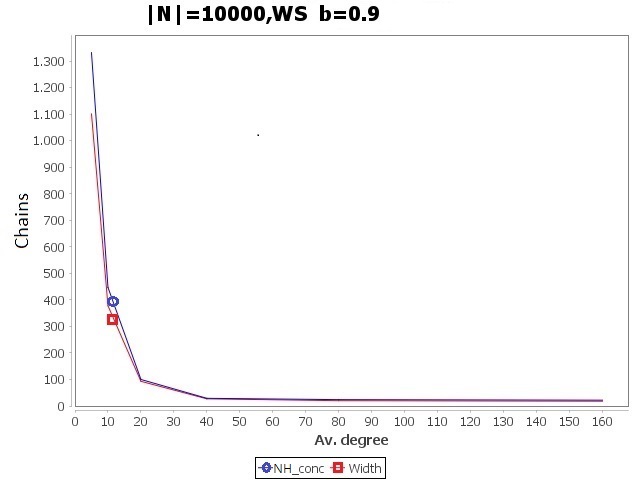}
\caption{The efficiency of our chain decomposition algorithm in WS model.}
\label{fig:Chart_WidthWatts09Comp10000}
\end{figure}

\begin{figure}[!ht]
\centering
 \includegraphics[scale=0.8]{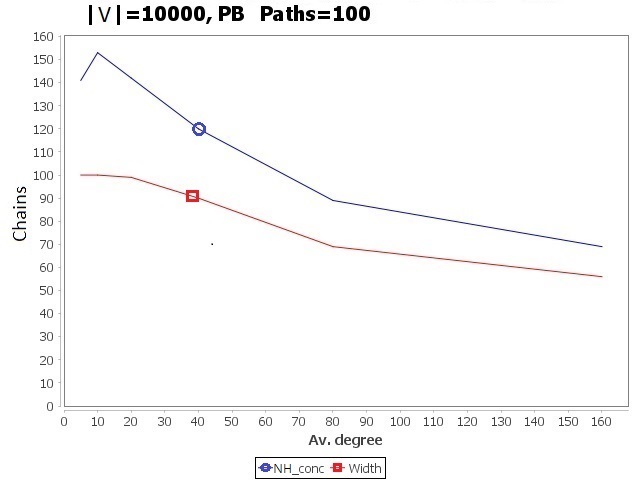}
\caption{The efficiency of our chain decomposition algorithm in the PB model.}
\label{fig:Chart_WidthPBMComp10000}
\end{figure}

\clearpage

\section{DAG Reduction for Faster Transitivity}\label{sec:HierTrans}
In this section, we show how to produce a linear time algorithm for the reduction of transitive edges, see Algorithm~\ref{alg:reduction}, which in combination with our fast chain decomposition, Algorithm~\ref{alg:NH_conc}, materialize the idea of a general-purpose reduction technique based on chain decomposition. This concept of reduction or abstraction of a DAG is useful in several applications beyond transitive closure. In~\cite{lionakis2019adventures} it was used to visualize (reduced) hierarchical graphs while displaying full reachability information. Therefore we state the following useful lemmas and Theorem~\ref{thm:EtrBound}:

\begin{lem}
\label{lemma:outgoing_nontransitive_edges}
Given a chain decomposition $D$ of a DAG $G=(V, E)$, each vertex $v_i \in V$, $0\leq i < |V|$, can have at most one outgoing non-transitive edge per chain.
\end{lem}
\begin{pf}
Given a graph $G(V,E)$, a decomposition $D(C_1, C_2, ..., C_{k_c})$ of G, and a vertex $v \in V$, assume vertex v has two outgoing edges, $(v,t_1)$ and $(v,t_2)$, and both $t_1$ and $t_2$ are in chain $C_i$. The vertices are in ascending topological order in the chain by definition. Assume $t_1$ has a lower topological rank than $t_2$. Thus, there is a path from $t_1$ to $t_2$, and accordingly a path from $v$ to $t_2$ through $t_1$. Hence, the edge $(v,t_2)$ is transitive. See Figure \ref{fig:transEdges_a}.  
\end{pf}

\begin{lem}
\label{lemma:incomming_nontransitive_edges}
Given any chain decomposition $D$ of a DAG $G=(V, E)$, each vertex $v_i \in V$, $0\leq i < |V|$, can have at most one incoming non-transitive edge per chain.
\end{lem}
\begin{pf}
Similar to the proof of Lemma \ref{lemma:outgoing_nontransitive_edges}, see Figure \ref{fig:transEdges_b}.  
\end{pf}

\begin{thm}
\label{thm:EtrBound}
Let $G=(V, E)$ be a DAG with width $w$. The non-transitive edges of $G$ are less than or equal to $width*|V|$, in other words $|E_{red}|=|E|-|E_{tr}|\leq width*|V|$. 
\end{thm}
\begin{pf}
Given any DAG $G$ and its width $w$, there is a chain decomposition of $G$ with $w$ number of chains. By Lemma \ref{lemma:outgoing_nontransitive_edges}, every vertex of G could have only one outgoing, non-transitive edge per chain. The same holds for the incoming edges, according to Lemma \ref{lemma:incomming_nontransitive_edges}.  Thus the non-transitive edges of $G$ are bounded by $width *|V|$. 
\end{pf}

\begin{figure}[ht!]
     \centering
     \begin{subfigure}[b]{0.25\textwidth}
         \centering
         \includegraphics[width=\textwidth]{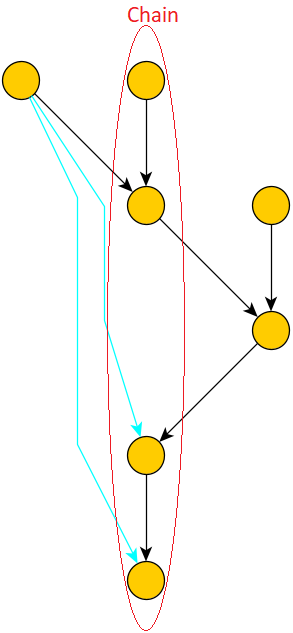}
         \caption{}
         \label{fig:transEdges_a}
     \end{subfigure}
     \hspace{2cm}
     \begin{subfigure}[b]{0.25\textwidth}
         \centering
         \includegraphics[width=\textwidth]{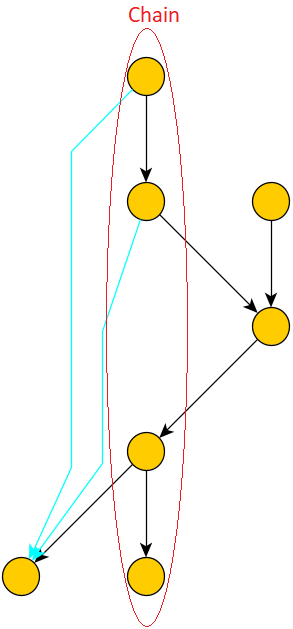}
         \caption{}
         \label{fig:transEdges_b}
     \end{subfigure}
        \caption{
        The light blue edges are transitive. (a) shows the outgoing transitive edges that end in the same chain. (b) shows the incoming transitive edges that start from the same chain.}
        \label{fig:transEdges}
\end{figure}

An interesting application of the above is that we can find a significantly large subset of $E_{tr}$ in linear time. Given any chain (or path) decomposition with $k_c$ chains, we can trace the vertices and their outgoing edges and keep the edges that point to the lowest point of each chain, rejecting the rest as transitive. We do the same for the incoming edges keeping the edges that come from the highest point (i.e., the vertex with the highest topological rank) of each chain.

Algorithm~\ref{alg:reduction}, is a linear time algorithm that reduces the outgoing transitive edges. The algorithm traverses all the vertices (outer loop) and utilizes a $k_c$-size array. For each vertex, it initializes the $k_c$-size array within the first inner loop. Then, it populates the array with the edges having the lowest rank for every chain during the second inner loop. Finally, it retrieves the reduced edge list from the $k_c$-size array in the third inner loop. This algorithm is simple and efficient but is not the only nor the most advanced way to reduce the graph in linear time using a chain decomposition. For example, there are several local heuristic choices that could be employed in order to find more transitive edges. To maintain our concise description approach to transitive closure, we refrain from elaborating on these details.

In this fashion we find a superset of $E_{red}$, call it $E_{red}'$, in linear time. Equivalently, we can find $E_{tr}'=E-E_{red}'$.  $E_{tr}'$ is a significantly large subset of $E_{tr}$ since $|E-E_{tr}'| = |E_{red}'|\leq k_c*|V|$.
Clearly, this approach can be used as a linear-time preprocessing step in order to substantially reduce the size of any DAG while keeping the same transitive closure as the original DAG $G$.  Consequently, this will speed up every transitive closure algorithm bounding the number of edges of any input graph, and the indegree and outdegree of every vertex by $k_c$.
For example, algorithms based on tree cover, see \cite{agrawal1989efficient,chen2005stack,trissl2007fast,wang2006dual},
are practical on sparse graphs and can be enhanced further with such a preprocessing step that removes transitive edges. Additionally, this approach may have practical applications in dynamic or hybrid transitive closure techniques: If one chooses to answer queries online by using graph traversal for every query, one could reduce the size of the graph with a fast (linear-time) preprocessing step that utilizes chains. Also, in the case of insertion/deletion of edges one could quickly decide if the edges to add or remove are transitive. Transitive edges do not affect the transitive closure, hence no updates are required. This could be practically useful in the dynamic insertion/deletion of edges.

\begin{algorithm}[!ht]
\caption{Reduction of outgoing edges.}
\label{alg:reduction}
\begin{algorithmic}[1]
\Procedure{Reduction}{$G,D$}\newline
\textbf{INPUT:} A DAG $G=(V,E)$, and the decomposition D of size k of G.
\State $incidentEdges[] \gets \mbox{new array of size k that holds edges.}$
\For{\textbf{each vertex:} $v \in G$}
 \For{\textbf{every outgoing edge $e(v,t)$ } }
  \State $chain \gets \mbox{The chain number of vertex } t.$
  \State $incidentEdges[chain] \gets \mbox{e}$
 \EndFor
 
 \For{\textbf{every outgoing edge $e(v,t)$ } }
  \State $chain \gets \mbox{The chain number of vertex } t.$
  \State $e'(v,t') \gets incidentEdges[chain]$
  \If{ $t'\mbox{ succeeds }t\mbox{ in the chain}$}    
   \State $incidentEdges[chain] \gets \mbox{e}$
  \EndIf
 \EndFor

 \State $reducedAdjList \gets \mbox{new empty list.}$
 \For{\textbf{every outgoing edge $e(v,t)$ } }
  \State $chain \gets \mbox{The chain number of vertex } t.$
  \State $e'(v,t') \gets incidentEdges[chain]$
  \If{$t' ==  t$}
   \State $reducedAdjList.\mbox{add}(t)$
  \EndIf
 \EndFor
\State $v.\mbox{adjTargetList} \gets reducedAdjList$
\EndFor
\EndProcedure
\end{algorithmic}
\end{algorithm}

\section{Reachability Indexing Scheme}  \label{sec:IndexingScheme}
In this section, we present an important application that uses a chain decomposition of a DAG. Namely, we solve the transitive closure problem by creating a reachability indexing scheme that is based on a chain decomposition and we evaluate it by running extensive experiments. Our experiments shed light on the interplay of various important factors as the density of the graphs increases.

Jagadish described a compressed transitive closure technique in 1990~\cite{Jagadish:1990:CTM:99935.99944} by applying an indexing scheme and simple path/chain decomposition techniques. His method uses successor lists and focuses on the compression of the transitive closure.
Simon~\cite{simon}, describes a technique similar to~\cite{Jagadish:1990:CTM:99935.99944}. His technique is based on computing a path decomposition, thus boosting the method presented in~\cite{goralcikova1979}. The linear time heuristic used by Simon is similar to the Chain Order Heuristic of~\cite{Jagadish:1990:CTM:99935.99944}.
A different approach is a graph structure referred to as path-tree cover introduced in~\cite{jin2008efficiently}, similarly, the authors utilize a path decomposition algorithm to build their labeling.

In the following subsections, we describe how to compute an indexing scheme in $O(|E_{tr}|+k_c*|E_{red}|)$ 
time, where $k_c$ is the number of chains (in any given chain decomposition) and $|E_{red}|$ is the number of non-transitive edges. Following the observations of Section~\ref{sec:HierTrans},
the time complexity of the scheme can be expressed as $O(|E_{tr}|+k_c*|E_{red}|) = O(|E_{tr}|+k_c*width*|V|)$ since $|E_{red}| \leq width*|V|$.  
Using an approach similar to Simon's~\cite{simon} our scheme creates arrays of indices to answer queries in constant time. The space complexity is $O(k_c*|V|)$.

 For our experiments, we utilize the chain decomposition approach of~\cite{kritikakis2022fast}, which produces smaller decompositions than previous heuristic techniques, without any considerable run-time overhead. Additionally, this heuristic, called NH\_conc, will perform better than any path decomposition algorithms as will be explained next. Thus the indexing scheme is more efficient both in terms of time and space requirements.
 Furthermore, the experimental work shows that, as expected, the chains rarely have the same length. Usually, a decomposition consists of a few long chains and several short chains. Hence, for most graphs, it is not even possible to have $|E_{red}|=width*|V|$, which assumes the worst case for the length of each chain. In fact, $|E_{red}|$ is usually much lower than that and the experimental results presented in  Tables \ref{table:IndexingSchemeResults5000} and \ref{table:IndexingSchemeResults10000} confirm this observation in practice.

Given a directed graph with cycles, we can find the strongly connected components (SCC) in linear time. Since any vertex is reachable from any other vertex in the same SCC (they form an equivalence class), all vertices in a SCC can be collapsed into a supernode. Hence, any reachability query can be reduced to a query in the resulting directed acyclic graph (DAG). This is a well-known step that has been widely used in many applications. Therefore, without loss of generality, we assume that the input graph to our method is a DAG. The following general steps describe how to compute the reachability indexing scheme:
\begin{enumerate}
  \item Compute a Chain decomposition
  \item Sort all Adjacency Lists
  \item Create an Indexing Scheme 
\end{enumerate}
In Step 1, we use our chain decomposition technique that runs in $O(|E|+c*l)$ time.
In Step 2, we sort all the adjacency lists in $O(|V|+|E|)$ time. Finally, we create an indexing scheme in
$O(|E_{tr}|+k_c*|E_{red}|)$ time and $O(k_c*|V|)$ space.  
Clearly, if the algorithm of Step 1 computes fewer chains then Step 3 becomes more efficient in terms of time and space.

\subsection{The Indexing Scheme}
Given any chain decomposition of a DAG $G$ with size $k_c$, an indexing scheme will be computed for every vertex that includes a pair of integers and an array of size $k_c$ of indexes. A small example is depicted in Figure \ref{IndexingSchemeExample}. The first integer of the pair indicates the node's chain and the second its position in the chain. For example, vertex 1 of Figure \ref{IndexingSchemeExample} has a pair $(1,1)$. This means that vertex 1 belongs to the $1st$ chain, and it is the $1st$ element in it.
Given a chain decomposition, we can easily construct the pairs in $O(|V|)$ time using a simple traversal of the chains.
Every entry of the $k_c$-size array represents a chain. The $i$-th cell represents the $i$-th chain. The entry in the $i$-th cell corresponds to the lowest point of the $i$-th chain that the vertex can reach. For example, the array of vertex 1 is $[1,2,2]$. The first cell of the array indicates that vertex 1 can reach the first vertex of the first chain (can reach itself, reflexive property). The second cell of the array indicates that vertex 1 can reach the second vertex of the second chain (There is a path from vertex 1 to vertex 7). Finally, the third cell of the array indicates that vertex 1 can reach the second vertex of the third chain.

Notice that we do not need the second integer of all pairs. If we know the chain a vertex belongs to, we can conclude its position using the array. We use this presentation to simplify the understanding of the users.

The process of answering a reachability query is simple. Assume, there is a source vertex $s$ and a target vertex $t$. To find if vertex $t$ is reachable from $s$, we first find the chain of $t$, and we use it as an index in the array of $s$. Hence, we know the lowest point of $t$'s chain vertex $s$ can reach. $s$ can reach $t$ if that point is less than or equal to $t$'s position, else it cannot.

\begin{figure}
\centering
\includegraphics[width=0.6\textwidth]{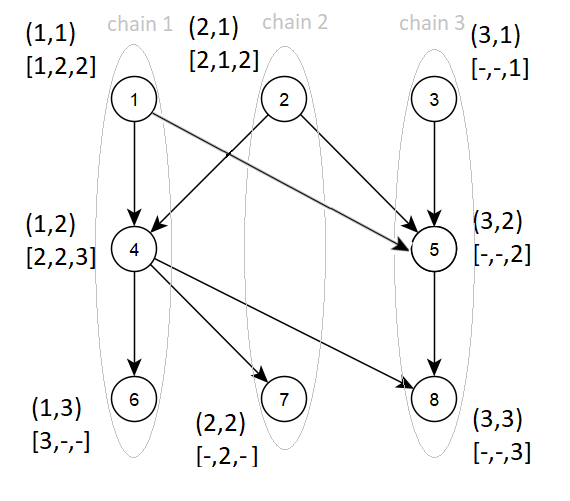}
\caption{An example of an indexing scheme.}
\label{IndexingSchemeExample}
\end{figure}

\subsection{Sorting the Adjacency lists}
Next, we use a linear time algorithm to sort all the adjacency lists of immediate successors in ascending topological order. See Algorithm~\ref{alg:SortingAdjList}. The algorithm maintains a stack for every vertex that indicates the sorted adjacency list. Then it traverses the vertices in reverse topological order, $(v_n,...,v_1)$. For every vertex $v_i$, $1\leq i\leq n$, it pushes $v_i$ into all immediate predecessors' stacks.
This step can be performed as a preprocessing step, even before receiving the chain decomposition. To emphasize its crucial role in the efficient creation of the indexing scheme, if the lists are not sorted then the second part of the time complexity would be $O(k_c*|E|)$ instead of $O(k_c*|E_{red}|)$.

\begin{algorithm}[!ht]
\caption{Sorting Adjacency lists}
\label{alg:SortingAdjList}
\begin{algorithmic}[1]
\Procedure{Sort}{$G,t$}\newline
 \textbf{INPUT:} A DAG $G=(V,E)$ 
\For{\textbf{each vertex:} $v_i \in G$}
\State $v_i\mbox{.stack  } \gets \mbox{new stack()}$
\EndFor
\For{\textbf{each vertex $v_i$ in reverse topological order} }
\For{\textbf{every incoming edge $e(s_j,v_i)$ } }
\State $s_j\mbox{.stack.push($v_i$)} $ 
\EndFor
\EndFor
\EndProcedure
\end{algorithmic}
\end{algorithm}

\begin{lem}
\label{lemma:SortingAdjList}
Algorithm \ref{alg:SortingAdjList} sorts the adjacency lists of immediate successors in ascending topological order, in linear time.
\end{lem}

\begin{proof}
Assume that there is a stack $(u_1, ..., u_n)$, $u_1$ is at the top of the stack. Assume that there is a pair $(u_j, u_k)$ in the stack, where $u_j$ has a bigger topological rank than $u_k$ and $u_j$ precedes $u_k$. This means that the for-loop examined $u_j$ before $u_k$.  Since the algorithm processes the vertices in reverse topological order, this is a contradiction. Vertex $u_j$ cannot precede $u_k$ if it were examined first by the for-loop. The algorithm traces all the incoming edges of every vertex. Therefore, it runs in linear time.
\end{proof}

\subsection{Creating the Indexing Scheme}
Now we present Algorithm \ref{alg:IndexingScheme} that constructs the indexing scheme. The first for-loop initializes the array of indexes. For every vertex, it initializes the cell that corresponds to its chain. The rest of the cells are initialized to infinity. The indexing scheme initialization is illustrated in Figure \ref{IndexingSchemeExample_init}. The dashes represent the infinite values. Notice that after the initialization, the indexes of all sink vertices have been calculated. Since a sink has no successors, the only vertex it can reach is itself.

The second for-loop builds the indexing scheme. It goes through vertices in descending topological order. For each vertex, it visits its immediate successors (outgoing edges) in ascending topological order and updates the indexes.
Suppose we have the edge $(v,s)$, and we have calculated the indexes of vertex $s$ ($s$ is an immediate successor of $v$). The process of updating the indexes of $v$ with its immediate successor, $s$, means that $s$ will pass all its information to vertex $v$. Hence, vertex $v$ will be aware that it can reach $s$ and all its successors. Assume the array of indexes of $v$ is $[a_1,a_2,...,a_{k_c}]$ and the array of $s$ is $[b_1,b_2,...,b_{k_c}]$. To update the indexes of $v$ using $s$, we merely trace the arrays and keep the smallest values. For every pair of indexes $(a_i, b_i)$, $0\leq i < k_c$, the new value of $a_i$ will be min\{$a_i$, $b_i$\}. This process needs $k_c$ steps.

\begin{lem}
\label{lemma:calc_indexes}
Given a vertex v and the calculated indexes of its successors, the while-loop of Algorithm \ref{alg:IndexingScheme} (lines 10-17) calculates the indexes of $v$ by updating its array with its non-transitive outgoing edges' successors.
\end{lem}
\begin{proof}
\label{proof:calc_indexes}
Updating the indexes of vertex $v$ with all its immediate successors will make $v$ aware of all its descendants. The while-loop of Algorithm \ref{alg:IndexingScheme} does not perform the update function for every direct successor. It skips all the transitive edges. Assume there is such a descendant $t$ and the transitive edge $(v,t)$. Since the edge is transitive, we know by definition that there exists a path from $v$ to $t$ with a length of more than 1. Suppose that the path is $(v,v_1,..,t)$. Vertex $v_1$ is a predecessor of $t$ and immediate successor of $v$. Hence it has a lower topological rank than $t$. Since, while-loop examines the incident vertices in ascending topological order, then vertex $t$ will be visited after vertex $v_1$. The opposite leads to a contradiction.
Consequently, for every incident transitive edge of $v$, the loop firstly visits a vertex $v_1$ which is a predecessor of $t$. Thus vertex $v$ will be firstly updated by $v_1$ and it will record the edge $(v,t)$ as transitive. Hence there is no reason to update the indexes of vertex $v$ with those of vertex $t$ since the indexes of $t$ will be greater than or equal to those of $v$.
\end{proof}

\begin{algorithm}[H]
\caption{Indexing Scheme}
\label{alg:IndexingScheme}
\begin{algorithmic}[1]
\Procedure{Create Indexing Scheme}{$G,T,D$}\newline
 \textbf{INPUT:} A DAG $G=(V,E)$, a topological sorting T of G, and the decomposition D of G.
\For{\textbf{each vertex:} $v_i \in G$}
\State $v_i\mbox{.indexes  } \gets \mbox{new table[size of D]}$
\State $v_i\mbox{.indexes.fill(} \infty \mbox{)}$
\State $ch\_no \gets v_i\mbox{'s chain index} $
\State $pos \gets v_i\mbox{'s chain position} $
\State $v_i\mbox{.indexes[ } ch\_no \mbox{ ]} \gets pos$
\EndFor
\For{\textbf{each vertex $v_i$ in reverse topological order} }
\For{\textbf{each adjacent target vertex $t$ of $v_i$ in ascending topological order} }
\State $t\_ch \gets \mbox{chain index of $t$} $
\State $t\_pos \gets \mbox{chain position of $t$} $
\If{ $t\_pos<v_i$.indexes[$t\_ch$]}
\Comment{ $(v_i,t)$ is not transitive}
\State $v_i\mbox{.updateIndexes}(t.indexes)$
\EndIf
\EndFor
\EndFor
\EndProcedure
\end{algorithmic}
\end{algorithm}

\begin{figure}
\centering
\includegraphics[width=0.6\textwidth]{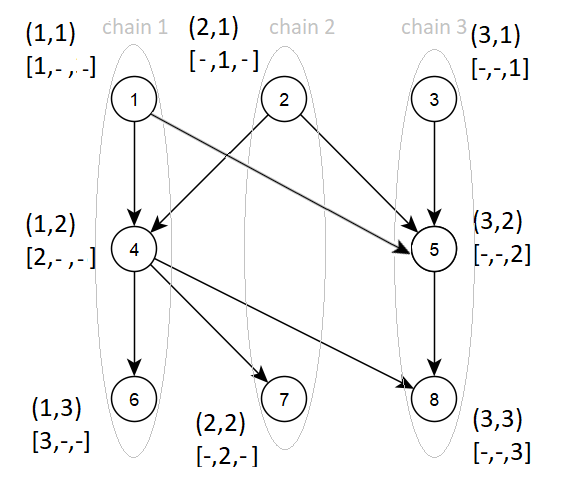}
\caption{Initialization of indexes.}
\label{IndexingSchemeExample_init}
\end{figure}

Combining the previous algorithms and results we have the following:

\begin{thm}
\label{thm:IndexingSchemeTime}
Let $G=(V,E)$  be a DAG. Algorithm \ref{alg:IndexingScheme} computes an indexing scheme for $G$ in $O(|E_{tr}|+k_c*|E_{red}|)$ time.
\end{thm}
\begin{proof}
In the initialization step, the indexes of all sink vertices have been computed as we described above. 
Taking vertices in reverse topological order, the first vertex we meet is a sink vertex. When the for-loop of line $9$ visits the first non-sink vertex, the indexes of its successors are computed (all its successors are sink vertices).
According to Lemma~\ref{lemma:calc_indexes}, we can calculate its indexes, ignoring the transitive edges. Assume the for-loop has reached vertex $v_i$ in the $i$th iteration, and the indexes of its successors are calculated. 
Following Lemma~\ref{lemma:calc_indexes}, we can calculate its indexes. Hence, by induction, we can calculate the indexes of all vertices, ignoring all $|E_{tr}|$ transitive edges in $O(|E_{tr}|+k_c*|E_{red}|)$ time.
\end{proof}

Since $O(|E_{tr}|+k_c*|E_{red}|)$ = $O(|E_{tr}|+k_c*w*|V|)$ $\leq O(|E_{tr}|+k_c^2*|V|)$, we conclude that the transitive closure (indexing scheme) of $G$ can be computed in parameterized linear time given any chain decomposition.
\begin{corollary}
    \label{cor:IndexingSchemeTime_kc}
Let $G=(V,E)$ be a DAG. Algorithm \ref{alg:IndexingScheme} can be used to compute an indexing scheme (transitive closure) of $G$ in parameterized linear time given any chain decomposition with $k_c$ chains.
\end{corollary}

As described in the introduction, a parameterized linear-time algorithm for computing the minimum number of chains was recently presented in~\cite{caceres2022sparsifying}.
Its time complexity is $O(k^3|V | + |E|)$
where $k$ is the minimum number of chains, which is equal to the width of $G$.
If we use this chain decomposition as input to Algorithm~\ref{alg:IndexingScheme} we have the following:
\begin{corollary}
    \label{cor:IndexingSchemeTime_width}
Let $G=(V,E)$ be a DAG. Algorithm \ref{alg:IndexingScheme} can compute an indexing scheme (transitive closure) of $G$ in parameterized linear time in terms of width.
\end{corollary}

\subsection{Experimental Results}
We conducted experiments using the same graphs of 5000 and 10000 nodes as we described in Section~\ref{sec:PathDec} that were produced by the four different models of Networkx~\cite{networkx} and the Path-Based model of~\cite{lionakis2021constant}. We computed a chain decomposition using the algorithm introduced in~\cite{kritikakis2022fast}, called NH\_conc, and created an indexing scheme using Algorithm \ref{alg:IndexingScheme}. For simplicity, we assume that the adjacency lists of the input graph are sorted, using Algorithm \ref{alg:SortingAdjList}, as a preprocessing step. 
We report our experimental results in Tables \ref{table:IndexingSchemeResults5000} and \ref{table:IndexingSchemeResults10000} for graphs with 5000 nodes and graphs with 10000 nodes, respectively.  

In theory, the phase of the indexing scheme creation needs $O(|E_{tr}|+k_c*|E_{red}|)$ time. However, the experimental results shown in the tables reveal some interesting (and expected) findings in practice:  As the average degree increases and the graph becomes denser, (a) the cardinality of $E_{red}$ remains almost stable; and (b) the number of chains decrease. The observation that the number of non-transitive edges, $E_{red}$, does not vary significantly as the average degree increases, implies that the number of transitive edges, $|E_{tr}|$, increases proportionally to the increase in the number of edges, since $(E_{tr}=E-E_{red})$. Since the algorithm merely traces in linear time the transitive edges, the growth of $|E_{tr}|$ affects the run time only linearly. As a result, the run time of our technique does not increase significantly as the the size (number of edges) of the input graph increases.  
In order to demonstrate this fact visually, we show the curves of the running time for the graphs of 10000 nodes produced by the ER model in Figure~\ref{fig:TimeErdosComp10000}. The $x$-axis shows the average degree of the nodes, i.e., the increasing density of the graphs.
The flat (blue line) represents the run time to compute the indexing scheme, and the curve (red line) represents the run time of the DFS-based algorithm for computing the transitive closure (TC). Clearly, the time of the DFS-based algorithm increases as the average degree (density) increases, while the time of the indexing scheme is a straight line almost parallel to the $x$-axis, i.e., remains almost constant with respect to the density. 
All models of Tables \ref{table:IndexingSchemeResults5000} and \ref{table:IndexingSchemeResults10000}  follow this pattern.

Apparently, there is a trade-off to consider when building an indexing scheme deploying our chain decomposition technique. The heuristic performs concatenations between paths. For every successful concatenation, the extra runtime overhead is $O(l)$, where $l$ is the longest path between the two concatenated paths. The unsuccessful concatenations do not cause any overhead.
Assume that we have a path decomposition, and then we perform chain concatenation.  
If there is no concatenation between two paths, the concatenation algorithm will run in linear time.

On the other hand, if there are concatenations, for each one of them, the cost is $O(l)$ time, but the savings in the indexing scheme creation is $\Theta(|V|)$ in space requirements and $\Theta (|E_{red}|)$ in time since every concatenation reduces the needed index size for every vertex by one. Hence, instead of computing a simple path decomposition (in linear time) the use of our chain concatenation procedure in order to create a more compact indexing scheme faster is always spreferred.
Another interesting and to some extent surprising observation that comes from the results of Tables \ref{table:IndexingSchemeResults5000} and \ref{table:IndexingSchemeResults10000} is that the transitive edges for almost all models of the graphs of 5000 and 10000 nodes with average degree above 20 are above 85\%, i.e.,  $|{E}_{{tr}}|/|{E}| \geq 85\%$, see the appropriate columns in both tables.  In some cases where the graphs are a bit denser, the percentage grows above 95\%.  This observation has important implications in designing practical algorithms for faster transitive closure computation in both the static and the dynamic case.

\begin{figure}[H]
\centering
\includegraphics[scale=0.8]{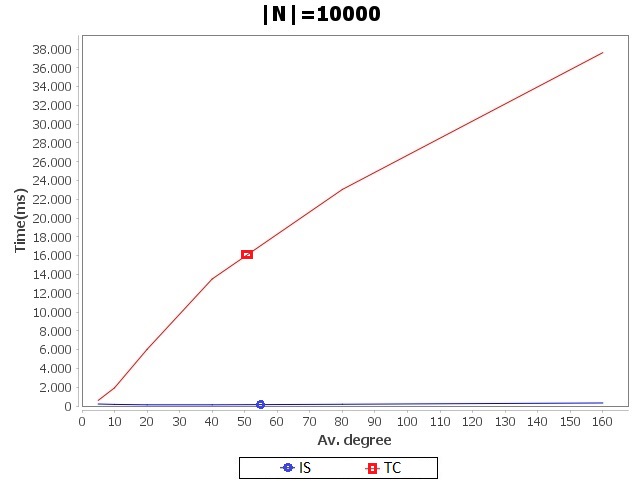}
\caption{Run time comparison between the Indexing Scheme (blue line) and TC (red line) for ER model on graphs of 10000 nodes, see Table \ref{table:IndexingSchemeResults10000}.}
\label{fig:TimeErdosComp10000}
\end{figure}

\newgeometry{top=1cm, bottom=2cm, left=2cm, right=2cm}
\begin{table}[H]
    \centering
    \begin{tabular}{|m{1.3cm}|| m{1.3cm}| m{1.3cm}| m{1cm}|m{1.5cm}| m{1.6cm}| m{1.4cm}|m{1cm}|m{1cm}|} 
\hline
\multicolumn{9}{|c|}{$|V|=5000$} \\
\hline
 Average Degree & Number of Chains & $|E_{tr}|$ & $|E_{red}|$ &$|E_{tr}|/|E|$ & $NH\_conc$ Time (ms) & Indexing Scheme Time (ms) & Total time (ms)  & TC\\ 
 \hline 
   & \multicolumn{8}{c|}{BA} \\
 \hline
 5 & 1630 & 8054 & 18921 & 0.32 & 3 & 101 & 104 & 137\\
 10 & 1055 & 28230 & 21670 & 0.57 & 12 & 79 & 91 & 333 \\
 20 & 664 & 75801 & 23799 & 0.76 & 6 & 54 & 60 & 638 \\
 40 & 335 & 180815 & 22504 & 0.89 & 10 & 48 & 58 & 1418\\ 
 80 & 207 & 382422 & 20854 & 0.95 & 122 & 118 & 240 & 3018\\
 160 & 163 & 770771 & 17660 & 0.98 & 25 & 107 & 132 & 5464\\
 \hline
   & \multicolumn{8}{c|}{ER} \\
 \hline
 5 & 923 & 3440 & 21466 & 0.14 & 6 & 67 & 73 & 172\\
 10 & 492 & 24761 & 25425 & 0.49 & 10 & 51 & 61 & 487 \\
 20 & 252 & 75312 & 24646 & 0.75 & 5 & 26 & 31 & 1079 \\
 40 & 139 & 175809 & 22634 & 0.89 & 46 & 51 & 97 & 2896\\ 
 80 & 70 & 378015 & 19435 & 0.95 & 16 & 50 & 66 & 5260\\
 160 & 38 & 769919 & 16843 & 0.98 & 98 & 138 & 236 & 8609\\
 \hline
    & \multicolumn{8}{c|}{WS, b=0.9} \\
 \hline
 5 & 687 & 7742 & 17258 & 0.30 & 13 & 71 & 84 & 393\\
 10 & 212 & 37992 & 12008 & 0.76 & 11 & 18 & 29 & 817 \\
 20 & 60 & 89272 & 10728 & 0.89 & 23 & 22 & 45 & 1530 \\
 40 & 25 & 186486 & 13514 & 0.93 & 47 & 45 & 92 & 3704\\ 
 80 & 20 & 386294 & 13706 & 0.97 & 115 & 103 & 218 & 6172\\
 160 & 17 & 787066 & 12934 & 0.98 & 253 & 207 & 460 & 9173\\
 \hline
    & \multicolumn{8}{c|}{WS, b=0.3} \\
 \hline
 5 & 9 & 18421 & 6579 & 0.74 & 11 & 8 & 19 & 910\\
 10 & 4 & 43505 & 6495 & 0.87 & 8 & 11 & 19 & 1107 \\
 20 & 4 & 93490 & 6510 & 0.93 & 18 & 18 & 36 & 2176 \\
 40 & 5 & 193416 & 6584 & 0.97 & 17 & 18 & 35 & 4753\\ 
 80 & 4 & 393348 & 6652 & 0.98 & 98 & 82 & 180 & 7949\\
 160 & 5 & 793430 & 6570 & 0.99 & 250 & 166 & 416 & 11757\\
 \hline
     & \multicolumn{8}{c|}{PB, Paths=70} \\
 \hline
 5 & 86 & 14155 & 10809 & 0.57 & 8 & 7 & 15 & 206\\
 10 & 101 & 36801 & 13102 & 0.74 & 7 & 12 & 19 & 313 \\
 20 & 107 & 84168 & 15419 & 0.85 & 7 & 15 & 22 & 890 \\
 40 & 93 & 181388 & 16988 & 0.91 & 49 & 216 & 265 & 2584\\ 
 80 & 73 & 376220 & 17303 & 0.96 & 128 & 163 & 291 & 4603\\
 160 & 51 & 758207 & 16566 & 0.98 & 55 & 141 & 196 & 9358\\
 \hline
 
\end{tabular}
    \caption{Experimental results for the indexing scheme for graphs of 5000 nodes.}
    \label{table:IndexingSchemeResults5000}
\end{table}

\begin{table}[H]
    \centering
    \begin{tabular}{|m{1.3cm}|| m{1.3cm}| m{1.3cm}| m{1cm}|m{1.5cm}| m{1.6cm}| m{1.4cm}|m{1cm}|m{1cm}|} 
\hline
\multicolumn{9}{|c|}{$|V|=10000$} \\
\hline
Average Degree & Number of Chains & $|E_{tr}|$ & $|E_{red}|$ &$|E_{tr}|/|E|$ & $NH\_conc$ Time (ms) & Indexing Scheme Time (ms) & Total time (ms) & TC\\ 
 \hline 
   & \multicolumn{8}{c|}{BA} \\
 \hline
 5 & 3341 & 14544 & 35431 & 0.29 & 7 & 278 & 285 & 441\\
 10 & 2159 & 53503 & 46397 & 0.54 & 14 & 231 & 245 & 1379 \\
 20 & 1264 & 147791 & 51809 & 0.74 & 15 & 218 & 233 & 3347 \\
 40 & 752 & 355854 & 52465 & 0.85 & 28 & 188 & 216 & 7700\\ 
 80 & 400 & 764926 & 48350 & 0.94 & 271 & 322 & 593 & 14632\\
 160 & 228 & 1560464 & 42967 & 0.97 & 81 & 264 & 345 & 24601\\
 \hline
   & \multicolumn{8}{c|}{ER} \\
 \hline
5 & 1837 & 5595 & 44401 & 0.11 & 12 & 200 & 212 & 600\\
 10 & 1003 & 44813 & 55366 & 0.45 & 9 & 161 & 170 & 1935 \\
 20 & 516 & 144276 & 55310 & 0.72 & 16 & 110 & 126 & 6031 \\
 40 & 271 & 347323 & 52620 & 0.87 & 25 & 101 & 126 & 13522\\ 
 80 & 139 & 749781 & 46666 & 0.94 & 40 & 145 & 185 & 23052\\
 160 & 72 & 1548153 & 39710 & 0.97 & 73 & 249 & 322 & 37613\\
 \hline
    & \multicolumn{8}{c|}{WS, b=0.9} \\
 \hline
 5 & 1332 & 13353 & 36647 & 0.27 & 12 & 175 & 187 & 1213\\
 10 & 447 & 74782 & 25218 & 0.75 & 9 & 53 & 62 & 3829 \\
 20 & 100 & 178930 & 21070 & 0.89 & 13 & 32 & 45 & 9279 \\
 40 & 29 & 373054 & 26946 & 0.93 & 24 & 60 & 84 & 13144\\ 
 80 & 24 & 771374 & 28626 & 0.96 & 266 & 247 & 513 & 25585\\
 160 & 22 & 1571957 & 28043 & 0.98 & 80 & 232 & 312 & 36507\\
 \hline
    & \multicolumn{8}{c|}{WS, b=0.3} \\
 \hline
 5 & 12 & 36816 & 13184 & 0.73 & 27 & 19 & 46 & 3468\\
 10 & 4 & 86804 & 13196 & 0.86 & 18 & 45 & 63 & 5063 \\
 20 & 4 & 186756 & 13244 & 0.93 & 10 & 42 & 52 & 12156 \\
 40 & 4 & 386751 & 13249 & 0.97 & 19 & 48 & 67 & 21055\\ 
 80 & 4 & 786840 & 13160 & 0.98 & 237 & 187 & 424 & 31016\\
 160 & 4 & 1586896 & 13104 & 0.99 & 62 & 167 & 229 & 40704\\
 \hline
     & \multicolumn{8}{c|}{PB, Paths=100} \\
 \hline
 5 & 125 & 8182 & 16810 & 0.33 & 12 & 16 & 28 & 240\\
 10 & 141 & 74182 & 25722 & 0.74 & 11 & 30 & 41 & 937 \\
 20 & 153 & 168839 & 30728 & 0.85 & 13 & 43 & 56 & 5015 \\
 40 & 142 & 363753 & 34606 & 0.91 & 27 & 78 & 105 & 13797\\ 
 80 & 120 & 756578 & 36918 & 0.96 & 56 & 142 & 198 & 27904\\
 160 & 89 & 1538101 & 36496 & 0.98 & 77 & 265 & 342 & 41235\\
 \hline
\end{tabular}
    \caption{Experimental results for the indexing scheme for graphs of 10000 nodes.}
    \label{table:IndexingSchemeResults10000}
\end{table}
\restoregeometry

\subsection{Speeding up the Method of Fulkerson}
\label{sec:FulkersonEvaluation}

In Section~\ref{sec:ExperimentsChains} we compute the minimum set of chains using the method of Fulkerson~\cite{Fulkerson}. As described, this method consists of three steps. The first step is the computation of the transitive closure, the second step involves the construction of a bipartite graph, and the third entails the computation of a maximal matching.
It has been a general belief for decades that the method of Fulkerson is not efficient (and thus not practically efficient) because of its dependency on the transitive closure problem. In this section, we show that the use of the indexing scheme (instead of the traditional transitive closure computation) in Fulkerson's method speeds up the computation of the width.

Specifically, the indexing scheme needs $O(k_c^2*|V| + |E_{tr}|)$ time, the construction of the bipartite graph has worst-case time complexity $O(|V|^2)$. The last step of maximal matching, requires $O(|E|*\sqrt{|V|})$ time, using the  algorithm by Hopcroft–Karp~\cite{DBLP:journals/siamcomp/HopcroftK73}. Hence, the implementation of the method requires $O(k_c^2*|V| + |V|^2 + |E|*\sqrt{|V|})$ time. 
In order to explore the practical significance of this we ran several experiments using  ER-model graphs.   Table~\ref{table:FulkersonEvaluation}, shows the runtime of each step of the Fulkerson method, and the final column shows the total runtime. Notice, that our transitive closure solution (indexing scheme), is not the most time-consuming step of this method, as believed in the past.  In fact it is the fastest step, and performs significantly better than the second step which is theoretically bounded by $O(|V|^2)$. 

This intriguing result can be explained by our earlier experimental results. As shown in Tables~\ref{table:IndexingSchemeResults5000} and~\ref{table:IndexingSchemeResults10000} the vast majority of edges are transitive. In fact, the number of non-transitive edges seem to be linear with respect to the number of nodes, (often only a few times the number of nodes). More precisely, our experiments show that the non-transitive edges are always less than six times the number of nodes. Hence, assuming that the number of non-transitive edges is a small constant times the number of nodes, then $|E_{red}|$ is $O(|V|)$.  Therefore, the expected runtime of the indexing scheme would be $O(k_c*|V| + |E_{tr}|)$.

Clearly, transitive closure solutions based on matrix multiplication cannot be faster than $O(|V|^2)$ since they perform computations on a two-dimensional adjacency matrix. Actually, it remains an open question whether it is possible to build a structure that requires $O(|V|^2)$ space and time while allowing us to answer queries in constant time.

The purpose of this sub-section is to show one application that directly benefits from the practical efficiency of our algorithms. By applying them to speed up Fulkerson's method, we challenge a prevailing belief about its effectiveness. It is not our intention to delve deeper into algorithms for minimum chain decomposition. Experimental results on optimal chain decomposition algorithms can be found in~\cite{rizzo2023chaining}.

\begin{table}[H]
    \centering
    \begin{tabular}{|m{1.4cm}|| m{1.6cm}| m{1.6cm}| m{2.5cm}|m{1cm}|} 
        \hline
        \multicolumn{5}{|c|}{Fulkerson Method} \\ 
        \hline
        Average Degree & Indexing Scheme Time (ms) & Bipartite Graph Time (ms) & Hopcroft–Karp Maximal Matching Time (ms) & Total Time (ms) \\ 
        \hline 
        & \multicolumn{4}{c|}{$|V|=5000$} \\ 
        \hline
        5 & 146 & 538 & 566 & 1250 \\
        10 & 98 & 1072 & 3228 &  4398 \\
        20 & 108 & 1223 & 5267 & 6598  \\
        40 & 42 & 916 & 8137 & 9095  \\
        80 & 165 & 1177 & 12265 & 13607 \\
        160 & 268 & 1129 & 10219 & 11616 \\
        \hline
       & \multicolumn{4}{c|}{$|V|=10000$} \\ 
        \hline
        5 & 466 & 1129 & 10219 & 11814 \\
        10 & 207 & 2794 & 10613 & 13614 \\
        20 & 241 & 3580 & 28918 & 32739 \\
        40 & 177 & 3925 & 48249 & 52351 \\
        80 & 321 & 5053 & 55315 & 60689 \\
        160 & 220 & 6107 & 69958 & 76285 \\
        \hline
    \end{tabular}
    \caption{Runtime evaluation of the Fulkerson Method utilizing the indexing scheme for ER graphs.}
    \label{table:FulkersonEvaluation}
\end{table}

\section{Conclusions and Extensions} \label{sec:conclusions}
In this paper, we extend ideas and experiments initiated in ~\cite{thesis, kritikakis2022fast, kritikakis2023fast}, providing a detailed and unified high-level overview. We describe the fastest way to compute a chain decomposition that is not merely a path decomposition and use it to implement transitive closure solutions. In particular, we focus on an indexing scheme that enables us to answer queries in constant time. We compute this scheme in parameterized linear time and space. Furthermore, we show how to reduce the graph in linear time given any path/chain decomposition removing transitive edges and how it can bolster transitive closure solutions. Our work diverges from mainstream approaches by exploring techniques that attempt to provide fast and practical solutions.
It is both theoretical and experimental, revealing crucial aspects of these problems. Our algorithms are applicable to real-world scenarios and can change the way we think about certain problems, as demonstrated by our experimental work on the method of Fulkerson, described in Section~\ref{sec:FulkersonEvaluation}.

Although our techniques were not developed for the dynamic case, where edges and nodes are added or deleted dynamically, the picture that emerges is very interesting. According to our experimental results, shown in Tables \ref{table:IndexingSchemeResults5000} and \ref{table:IndexingSchemeResults10000}, the overwhelming majority of edges in a DAG are transitive.  
The insertion or deletion of a transitive edge clearly requires a constant time update since it does not affect transitivity, and can be detected in constant time. On the other hand, the insertion or removal of a non-transitive edge may require a minor or major recomputation to reestablish a correct chain decomposition.  Similarly, since the nodes of the DAG are topologically ordered, the insertion of an edge that goes from a high node to a low node signifies that the SCCs of the graph have changed, perhaps locally.  However, even if the insertion/deletion of new nodes/edges causes significant changes in the reachability index (transitive closure) one can simply recompute a chain decomposition in linear or almost linear time, and then recompute the reachability scheme in parameterized linear time, $O(|E_{tr}|+k_c*|E_{red}|)$, and $O(k_c*|V|)$ space, which is still very efficient in practice.  For a very recent comparison of practical fully dynamic transitive closure techniques see~\cite{DBLP:journals/corr/abs-2002-00813}. 
We plan to work on the problems that arise in the computation of dynamic path/chain decomposition and reachability indexes in the future.


\appendix

 \bibliographystyle{elsarticle-num} 
 \bibliography{bibliography}




\end{document}